\documentclass[letterpaper,11pt]{article}
\usepackage{microtype}
\usepackage[letterpaper,margin=1in]{geometry}
\usepackage[UKenglish]{babel}
\usepackage[numbers,sort&compress]{natbib}
\usepackage{color}
\usepackage{doi}
\usepackage{latexsym}
\usepackage{amsmath}
\usepackage{amssymb}
\usepackage{amsthm}
\usepackage{thmtools}
\usepackage{thm-restate}
\usepackage[all,warning]{onlyamsmath}
\usepackage{enumitem}
\usepackage{framed}
\usepackage{xparse}
\usepackage{mathtools}
\usepackage{graphicx}
\usepackage[basic]{complexity}
\usepackage[small]{caption}
\usepackage{booktabs}
\usepackage{hyperref}
\usepackage{algpseudocode}

\theoremstyle{plain}
\newtheorem{theorem}{Theorem}
\newtheorem{lemma}[theorem]{Lemma}

\theoremstyle{definition}

\theoremstyle{remark}
\newtheorem{remark}[theorem]{Remark}

\setlist[itemize]{label=--}
\setlist[enumerate]{label=(\arabic*),labelindent=\parindent,leftmargin=*}

\DeclarePairedDelimiter\braces{\{}{\}}
\NewDocumentCommand\set{O{}mg}{\ensuremath{\braces[#1]{#2\IfNoValueTF{#3}{}{\,:\,#3}}}}

\DeclareMathOperator{\dist}{dist}

\DeclareMathOperator{\mypolylog}{polylog}

\newcommand{\N}{\mathbb{N}}

\newclass{\lcl}{LCL}
\newclass{\LLL}{LLL}
\newclass{\local}{LOCAL}

\newcommand{\A}{\ensuremath{\mathcal A}}

\newcommand{\namedref}[2]{\hyperref[#2]{#1~\ref*{#2}}}
\newcommand{\sectionref}[1]{\namedref{Section}{#1}}

\newcommand{\theoremref}[1]{\namedref{Theorem}{#1}}

\newcommand{\figureref}[1]{\namedref{Figure}{#1}}
\newcommand{\lemmaref}[1]{\namedref{Lemma}{#1}}
\newcommand{\tableref}[1]{\namedref{Table}{#1}}

\newcommand{\MCOUNT}{\mbox{\sc count}}

\newcommand{\MPOWK}{\mbox{\sc pow}}
\newcommand{\MEXP}{\mbox{\sc exp}}
\newcommand{\MROOTK}{\mbox{\sc root}}

\newcommand{\MLOG}{\mbox{\sc log}}

\newcommand{\dirup}{\ensuremath{\mathsf U}}
\newcommand{\dirdown}{\ensuremath{\mathsf D}}
\newcommand{\dirleft}{\ensuremath{\mathsf L}}
\newcommand{\dirright}{\ensuremath{\mathsf R}}
\newcommand{\error}{\ensuremath{\mathsf E}}

\newenvironment{myabstract}
{\list{}{\listparindent 1.5em%
		\itemindent    \listparindent
		\leftmargin    1cm
		\rightmargin   1cm
		\parsep        0pt}%
	\item\relax}
{\endlist}

\newenvironment{mycover}
{\list{}{\listparindent 0pt
		\itemindent    \listparindent
		\leftmargin    1cm
		\rightmargin   1cm
		\parsep        0pt}%
	\raggedright
	\item\relax}
{\endlist}

\newenvironment{program}
{\begin{oframed}\begin{algorithmic}}
{\end{algorithmic}\end{oframed}}

\newcommand{\myemail}[1]{\,$\cdot$\, {\small #1}}
\newcommand{\myaff}[1]{\,$\cdot$\, {\small #1}\par\medskip}

\hypersetup{
	colorlinks=true,
	linkcolor=black,
	citecolor=black,
	filecolor=black,
	urlcolor=[rgb]{0,0.1,0.5},
	pdftitle={New classes of distributed time complexity},
	pdfauthor={Alkida Balliu, Juho Hirvonen, Janne H. Korhonen, Tuomo Lempi\"ainen, Dennis Olivetti, Jukka Suomela}
}

\begin{document}

\begin{mycover}
	{\huge\bfseries\boldmath New Classes of Distributed Time Complexity \par}
	\bigskip
	\bigskip

	\textbf{Alkida Balliu}
	\myemail{alkida.balliu@aalto.fi}
	\myaff{Aalto University, Gran Sasso Science Institute, and  Institut de Recherche en Informatique Fondamentale}

	\textbf{Juho Hirvonen}
	\myemail{juho.hirvonen@cs.uni-freiburg.de}
	\myaff{University of Freiburg, IRIF, CNRS, and University Paris Diderot}

	\textbf{Janne H.\ Korhonen}
	\myemail{janne.h.korhonen@aalto.fi}
	\myaff{Aalto University}

	\textbf{Tuomo Lempi\"ainen}
	\myemail{tuomo.lempiainen@aalto.fi}
	\myaff{Aalto University}

	\textbf{Dennis Olivetti}
	\myemail{dennis.olivetti@aalto.fi}
	\myaff{Aalto University, Gran Sasso Science Institute, and  Institut de Recherche en Informatique Fondamentale}

	\textbf{Jukka Suomela}
	\myemail{jukka.suomela@aalto.fi}
	\myaff{Aalto University}

\end{mycover}

\medskip
\begin{myabstract}
	\noindent\textbf{Abstract.}
	A number of recent papers -- e.g.\ Brandt et al.\ (STOC 2016), Chang et al.\ (FOCS 2016), Ghaffari \& Su (SODA 2017), Brandt et al.\ (PODC 2017), and Chang \& Pettie (FOCS 2017) -- have advanced our understanding of one of the most fundamental questions in theory of distributed computing: what are the possible time complexity classes of $\lcl$ problems in the $\local$ model? In essence, we have a graph problem $\Pi$ in which a solution can be \emph{verified} by checking all radius-$O(1)$ neighbourhoods, and the question is what is the smallest $T$ such that a solution can be \emph{computed} so that each node chooses its own output based on its radius-$T$ neighbourhood. Here $T$ is the distributed time complexity of $\Pi$.

	The time complexity classes for deterministic algorithms in bounded-degree graphs that are known to exist by prior work are $\Theta(1)$, $\Theta(\log^* n)$, $\Theta(\log n)$, $\Theta(n^{1/k})$, and $\Theta(n)$. It is also known that there are two gaps: one between $\omega(1)$ and $o(\log \log^* n)$, and another between $\omega(\log^* n)$ and $o(\log n)$. It has been conjectured that many more gaps exist, and that the overall time hierarchy is relatively simple -- indeed, this is known to be the case in restricted graph families such as cycles and grids.

	We show that the picture is much more diverse than previously expected. We present a general technique for engineering $\lcl$ problems with numerous different deterministic time complexities, including $\Theta( \log^{\alpha} n )$ for any $\alpha \ge 1$, $2^{\Theta( \log^{\alpha} n )}$ for any $\alpha \le 1$, and $\Theta(n^{\alpha})$ for any $\alpha < 1/2$ in the high end of the complexity spectrum, and $\Theta( \log^{\alpha} \log^* n )$ for any $\alpha \ge 1$, $\smash{2^{\Theta( \log^{\alpha} \log^* n )}}$ for any $\alpha \le 1$, and $\Theta((\log^* n)^{\alpha})$ for any $\alpha \le 1$ in the low end of the complexity spectrum; here $\alpha$ is a positive rational number.
\end{myabstract}

\thispagestyle{empty}
\setcounter{page}{0}
\newpage

\section{Introduction}\label{sec:introduction}

In this work, we show that the landscape of distributed time complexity is much more diverse than what was previously known. We present a general technique for constructing distributed graph problems with a wide range of different time complexities. In particular, our work answers many of the open questions of Chang and Pettie~\cite{chang17hierarchy}, and disproves one of their conjectures.

\subsection{\texorpdfstring{\boldmath $\local$}{LOCAL} model}

We explore here one of the standard models of distributed computing, the $\local$ model~\cite{peleg00distributed,linial92locality}. In this model, we say that a graph problem (e.g., graph colouring) is solvable in time $T$ if each node can output its own part of the solution (e.g., its own colour) based on its radius-$T$ neighbourhood. We focus on deterministic algorithms -- even though most of our results have direct randomised counterparts -- and as usual, we assume that each node is labelled with an $O(\log n)$-bit unique identifier. We give the precise definitions in Section~\ref{sec:prelim}.

\subsection{\texorpdfstring{\boldmath $\lcl$}{LCL} problems}

The most important family of graph problems from the perspective of the $\local$ model is the class of $\lcl$ problems~\cite{naor95what}. Informally, $\lcl$ problems are graph problems that can be solved in constant time with a \emph{nondeterministic} algorithm in the $\local$ model, and the key research question is, what is the time complexity of solving $\lcl$ problems with \emph{deterministic} algorithms. Examples of $\lcl$ problems include the problem of finding a proper vertex colouring with $k$ colours: if you guess a solution nondeterministically, you can easily verify it with a constant-time distributed algorithm by having each node check its constant-radius neighbourhood. As usual, we will focus on bounded-degree graphs. We give the precise definitions in Section~\ref{sec:prelim}.

\subsection{State of the art}\label{ssec:prior}

Already in the 1990s, it was known that there are $\lcl$ problems with time complexities $O(1)$, $\Theta(\log^* n)$, and $\Theta(n)$ on $n$-node graphs~\cite{linial92locality,cole86deterministic}. It is also known that these are the only possibilities in the case of cycles and paths~\cite{naor95what}. For example, the problem of finding a $2$-colouring of a path is inherently global, requiring time $\Theta(n)$, while the problem of finding a $3$-colouring of a path can be solved in time $\Theta(\log^* n)$.

While some cases (e.g., oriented grids) are now completely understood~\cite{grid-lcl}, the case of general graphs is currently under active investigation. $\lcl$ problems with deterministic time complexities of $\Theta(\log n)$~\cite{brandt16lll,chang16exponential,ghaffari17distributed} and $\Theta(n^{1/k})$ for all $k$~\cite{chang17hierarchy} have been identified only very recently. It was shown by Chang et al.\ that there are no $\lcl{}$ problems with complexities between $\omega(\log^* n)$ and $o(\log n)$~\cite{chang16exponential}. Classical symmetry breaking problems like maximal matching, maximal independent set, $(\Delta+1)$-colouring, and $(2\Delta-1)$-edge colouring have complexity $\Theta(\log^* n)$~\cite{barenboim16sublinear,barenboim14distributed,panconesi01simple,fraigniaud16local}. Some classical problems are now also known to have intermediate complexities, even though tight bounds are still missing: $\Delta$-colouring and $(2\Delta-2)$-edge colouring require $\Omega(\log n)$ rounds~\cite{brandt16lll,chang18complexity}, and can be solved in time $O(\mypolylog n)$~\cite{panconesi95delta}. Some gaps have been conjectured; for example, Chang and Pettie~\cite{chang17hierarchy} conjecture that there are no problems with complexity between $\omega(n^{1/(k+1)})$ and $o(n^{1/k})$. See Table~\ref{tab:soa} for an overview of the state of the art.

\begin{table}
	\centering
	\begin{tabular}{@{}lll@{}}
		\toprule
		Complexity & Status & Reference \\
		\midrule
		$O(1)$ & exists & trivial \\
		$\omega(1)$, $o(\log \log^* n)$ & does not exist & \cite{naor95what} \\
		$\Omega(\log \log^* n)$, $o(\log^* n)$ & ? & \\
		$\Theta(\log^* n)$ & exists & \cite{linial92locality,cole86deterministic} \\
		$\omega(\log^* n)$, $o(\log n)$ & does not exist & \cite{chang16exponential} \\
		$\Theta(\log n)$ & exists & \cite{brandt16lll,chang16exponential,ghaffari17distributed} \\
		$\omega(\log n)$, $n^{o(1)}$ & ? & \\
		$\Theta(n^{1/k})$ & exists & \cite{chang17hierarchy} \\
		$\Theta(n)$ & exists & trivial \\
		\bottomrule
	\end{tabular}
	\caption{State of the art: prior results on the existence of $\lcl$ problems in different complexity classes.}\label{tab:soa}
\end{table}

The picture changes for randomised algorithms, especially in the region between $\omega(\log^* n)$ and $o(\log n)$. In the lower end of this region, it is known that there are no \lcl{}s with randomised complexity between $\omega(\log^* n)$ and $o(\log \log n)$~\cite{chang16exponential}, but e.g.\ sinkless orientations have a randomised complexity $\Theta(\log \log n)$~\cite{brandt16lll,ghaffari17distributed}, and it is known that no \lcl{} problem belongs to this complexity class in the deterministic world. In the higher end of this region, it is known that all \lcl{}s solvable in time $o(\log n)$ can be solved in time $T_{\LLL}(n)$, the time it takes to solve a relaxed variant of the Lov\'asz local lemma~\cite{chang17hierarchy}. The current best algorithm gives $T_{\LLL}(n) = 2^{O(\sqrt{\log \log n})}$~\cite{fischer17sublogarithmic}.

So far we have discussed the complexity of \lcl{} problems in the strict classical sense, in graphs of maximum degree $\Delta = O(1)$. Many of these problems have been studied also in the case of a general~$\Delta$. The best deterministic algorithms for maximal independent set and $(\Delta+1)$-colouring run in time $2^{O(\sqrt{\log n})}$~\cite{panconesi96decomposition}. Maximal matching can be solved in time $O(\log^2 \Delta \log n)$~\cite{fischer17improved} and $(2\Delta-1)$-edge colouring in time $O(\log^8 \Delta \log n)$~\cite{fischer17deterministic}. Corresponding randomised solutions are exponentially faster: $O(\log \Delta) + 2^{O(\sqrt{\log \log n})}$ rounds for maximal independent set~\cite{ghaffari16improved}, $O(\sqrt{\log \Delta}) + 2^{O(\sqrt{\log \log n})}$ for $(\Delta+1)$-colouring~\cite{harris16sublogarithmic}, $O(\log \Delta + \log^3 \log n)$ for maximal matching~\cite{fischer17improved}, and $2^{O(\sqrt{\log \log n})}$ for $(2\Delta-1)$-edge colouring~\cite{fischer17deterministic,elkin15edgecoloring}. Some lower bounds are known in the general case: maximal independent set and maximal matching require time $\Omega(\min \{ \log \Delta / \log \log \Delta, \sqrt{\log n / \log \log n} \})$~\cite{kuhn16local}.

\subsection{Contributions}

Based on the known results related to $\lcl$ problems, it seemed reasonable to conjecture that there might be only three distinct non-empty time complexity classes below $n^{o(1)}$, namely $O(1)$, $\Theta(\log^* n)$, and $\Theta(\log n)$. There are very few candidates of $\lcl$ problems that might have any other time complexity, and in particular the gap between $\Omega(\log \log^* n)$ and $o(\log^* n)$ seemed to be merely an artefact of the current Ramsey-based proof techniques (see Chang and Pettie~\cite{chang17hierarchy} for more detailed discussion on this region).

Our work changes the picture completely: we show how to construct \emph{infinitely many} $\lcl$ problems for the regions in which the existence of any problems was an open question. We present a \emph{general technique} that enables us to produce time complexities of the form $f(\log^* n)$ and $f(n)$ for a wide range of functions $f$, as long as $f$ is sublinear and at least logarithmic. See Table~\ref{tab:contrib} for some examples of time complexities that we can construct with our technique.

\begin{table}[b!]
	\newcommand{\new}{this work}
	\newcommand{\gap}{$\mspace{5mu}\bigg|$ \emph{gap}}
	\centering
	\begin{tabular}{@{}lll@{\qquad\qquad}lll@{}}
		\toprule
		``Low'' & & & ``High'' \\
		\midrule
		$O(1)$ & & trivial & $\Theta(\log^* n)$ & & \cite{linial92locality,cole86deterministic} \\[3pt]
		\gap & & \cite{naor95what,chang17hierarchy} & \gap & & \cite{chang16exponential} \\[11pt]
		$\Theta( \log \log^* n )$ & & \new & $\Theta(\log n)$ & & \cite{brandt16lll,chang16exponential,ghaffari17distributed}  \\[2pt]
		$\Theta( \log^{r/s} \log^* n )$& $r/s \ge 1$ & \new & $\Theta( \log^{r/s} n )$& $r/s \ge 1$ & \new \\[2pt]
		$\smash{2^{\Theta( \log^{r/s} \log^* n )}}$& $r/s \le 1$ & \new & $2^{\Theta( \log^{r/s} n )}$& $r/s \le 1$ &\new \\[2pt]
		& & & $\Theta(n^{1/s})$ & & \cite{chang17hierarchy} \\[2pt]
		$\Theta((\log^* n)^{r/s})$ & $r/s \le 1$ & \new & $\Theta(n^{r/s})$ & $r/s < 1/2$ & \new \\[2pt]
		$\Theta(\log^* n)$ & & \cite{linial92locality,cole86deterministic} & $\Theta(n)$ & & trivial \\
		\bottomrule
	\end{tabular}
	\caption{Our contributions: examples of time complexity classes that are now known to contain an \lcl{} problem. The integers $r$ and $s$ are positive constants.}\label{tab:contrib}
\end{table}

The table also highlights another surprise: the structure of ``low'' complexities below $O(\log^* n)$ and the structure of ``high'' complexities above $\omega(\log^* n)$ look now very similar.

\subsection{Proof ideas}

On a high level, we start by defining a simple model of computation, called a \emph{link machine} here. We emphasise that link machines are completely unrelated to distributed computing; they are simply a specific variant of the classical register machines. A link machine has $O(1)$ registers that can hold unbounded positive natural numbers, and a finite \emph{program} (sequence of instructions). The machine supports the following instructions: resetting a register to $1$, addition of two registers, comparing two registers for equality, and skipping operations based on the result of a comparison.

We say that a link machine $M$ has \emph{growth} $g\colon \N \to \N$ if the following holds: if we reset all registers to value $1$, and then run the program of the machine $M$ repeatedly for $\ell$ times, then the \emph{maximum} of the register values is $g(\ell)$. For example, the following link machine has a growth $g(\ell) = \Theta(\ell^2)$:
\begin{program}
	\State{$x \gets x + 1$}
	\State{$y \gets y + x$}
\end{program}
Now assume that we have the following ingredients:
\begin{enumerate}[noitemsep]
	\item A link machine $M$ of growth $g$.
	\item An $\lcl$ problem $\Pi$ for directed cycles, with a time complexity $T$.
\end{enumerate}
We show how to construct a new $\lcl$ problem $\Pi_M$ in which the ``relevant'' instances are graphs $G$ with the following structure:
\begin{itemize}[noitemsep]
	\item There is a directed cycle $C$ in which we need to solve the original problem $\Pi$.
	\item The cycle is augmented with an additional structure of multiple layers of ``shortcuts'', and the lengths of the shortcuts correspond to the values of the registers of machine $M$.
\end{itemize}
Therefore if we take $\ell$ steps away from cycle $C$, we will find shortcuts of length $g(\ell)$. In particular, if two nodes $u$ and $v$ are within distance $\ell \cdot g(\ell)$ from each other along cycle $C$, we can reach from $u$ to $v$ in $\Theta(\ell)$ steps along graph $G$.

In essence, we have compressed the distances and made problem $\Pi_M$ easier to solve than $\Pi$, in a manner that is controlled precisely by function $g$. For example, if $g(\ell) = \Theta(\ell^2)$, then distance $\ell \cdot g(\ell) = \Theta(\ell^3)$ along $C$ corresponds to distance $\Theta(\ell)$ in graph $G$. If $\Pi$ had a time complexity of $T = \Theta(\log^* n)$, we obtained a problem $\Pi_M$ with a time complexity of $T_M = \Theta((\log^* n)^{1/3})$.

Notice that these results could not be achieved by just adding shortcuts of length $g(n)$ directly onto every node of the cycle, since the lengths of the shortcuts would not be locally checkable, and moreover, it would not be true that a node can reach every other node within a certain distance.

\subsection{Some technical details}

Plenty of care is needed in order to make sure that
\begin{itemize}
	\item $\Pi_M$ is indeed a well-defined $\lcl$ problem: feasible solutions can be verified by checking the $O(1)$-radius neighbourhoods,
	\item $\Pi_M$ is solvable in time $O(T_M)$ also in arbitrary bounded-degree graphs and not just in ``relevant'' instances that have the appropriate structure of a cycle plus shortcuts,
	\item there is no way to cheat and solve $\Pi_M$ in time $o(T_M)$.
\end{itemize}
There is a fine balance between these goals: to make sure $\Pi_M$ is solvable efficiently in adversarial instances, we want to modify the definition so that for unexpected inputs it is permitted to produce the output ``this is an invalid instance'', but this easily opens loopholes for cheating -- what if all nodes always claim that the input is invalid?

We address these issues with the help of ideas from \emph{locally checkable proofs} and \emph{proof labelling schemes} \cite{goos16lcp,korman10proof}, both for inputs and for outputs:
\begin{itemize}
	\item Locally checkable inputs: Relevant instances carry a locally checkable proof. If the instance is not relevant, it can be detected locally.
	\item Locally checkable outputs: If the algorithm claims that the input is invalid, it also has to prove it. If the proof is wrong, it can be detected locally.
\end{itemize}
In essence, we define $\Pi_M$ so that the algorithm has two possibilities in all local neighbourhoods: solve $\Pi$ or prove that the input is invalid. This requirement can be now encoded as a bona fide $\lcl$ problem.

As a minor twist, we have slightly modified the above scheme so that we replace impossibility of cheating by hardness of cheating. Our $\lcl$ problem is designed so that an algorithm could, in principle, construct a convincing proof that claims that the input is invalid (at least for some valid inputs). However, to avoid detection, the algorithm would need to spend $\Omega(T_M)$ time to construct such a proof -- in which case the algorithm could equally well solve the original problem directly.

\subsection{Significance}

The complexity classes and the gaps in the time hierarchy of $\lcl$ problems have recently played a key role in the field of distributed computing. The classes have served as a source of inspiration for algorithm design (e.g.\ the line of research related to the sinkless orientation problem \cite{brandt16lll,chang16exponential,ghaffari17distributed} and the follow-up work \cite{ghaffari17degree-splitting} that places many other problems in the same complexity class), and the gaps have directly implied non-trivial algorithmic results (e.g.\ the problem of $4$-colouring $2$-dimensional grids \cite{grid-lcl}). The recently identified gaps \cite{grid-lcl,chang16exponential,chang17hierarchy,fischer17sublogarithmic} have looked very promising; it has seemed that a complete characterisation of the $\lcl$ complexities might be within a reach of the state of the art techniques, and the resulting hierarchy might be sparse and natural.

In essence, our work shows that the free lunch is over. The deterministic $\lcl$ complexities in general bounded-degree graphs do not seem to provide any further gaps that we could exploit. Any of the currently known upper bounds might be tight. To discover new gaps, we will need to restrict the setting further, e.g.\ by studying restricted graph families such as grids and trees \cite{grid-lcl,chang17hierarchy}, or by focusing on restricted families of $\lcl$ problems. Indeed, this is our main open question: what is the broadest family of $\lcl$ problems that contains the standard primitives (e.g., colourings and orientations) but for which there are large gaps in the distributed time hierarchy?

\section{Preliminaries}\label{sec:prelim}

Let us first fix some terminology. We work with directed graphs $G=(V,E)$, which are always assumed to be simple, finite, and connected. We denote the number of nodes by $n = |V|$. The number of edges on a shortest path from node~$v$ to node~$u$ is denoted $\dist(v,u)$. A \emph{labelling} of a graph~$G$ is a mapping $l\colon V \to \Sigma$. Given a labelled graph $(G,l)$, the \emph{radius-$T$ neighbourhood} of a node~$v$ consists of the subgraph~$G_{v,T}=(V_{v,T},E_{v,T})$, where $V_{v,T} = \set{u\in V}{\dist(v,u) \le T}$ and $E_{v,T} = \set{(u,w)\in E}{{\dist(v,u) \le T} \text{ and } {\dist(v,w) \le T}}$, as well as the restriction $l\restriction_{V_{v,T}}\colon V_{v,T} \to \Sigma$ of the labelling. The set of natural numbers is $\N = \set{0,1,\ldots}$.

\subsection{Model of computation}

Our setting takes place in the $\local$ model~\cite{peleg00distributed,linial92locality} of distributed computing. We have a graph~$G=(V,E)$, where each node $v \in V$ is a computational unit and all nodes run the same \emph{deterministic} algorithm~$\A$. We work with \emph{bounded-degree graphs}; hence $\A$ can depend on an upper bound~$\Delta$ for the maximum degree of $G$.

Initially, the nodes are not aware of the graph topology -- they can learn information about it by communicating with their neighbours. To break symmetry, nodes have access to $O(\log n)$-bit unique identifiers, given as a labelling. We will also assume that the nodes are given as input the number~$n$ of nodes (for most of our results, e.g.\ a polynomial upper bound on $n$ is sufficient). In addition, nodes can be given a task-specific local \emph{input labelling}. We will often refer to directed edges, but for our purposes the directions are just additional information that is encoded in the input labelling. We emphasise that the directions of the edges do not affect communication; they are just additional information that the nodes can use.

The communication takes place in synchronous communication rounds. In each round, each node $v \in V$
\begin{enumerate}[noitemsep]
  \item sends a message to each of its neighbours,
  \item receives a message from each of its neighbours,
  \item performs local computation based on the received messages.
\end{enumerate}
Each node~$v$ is required to eventually halt and produce its own local output. We do not limit the amount of local computation in each round, nor the size of messages; the only resource of interest is the number of communication rounds until all the nodes have halted.

Note that in $T$ rounds of communication, each node can gather all information in its radius-$T$ neighbourhood, and hence a $T$-round algorithm is simply a mapping from radius-$T$ neighbourhoods to local outputs.

\subsection{Graph problems}

In the framework of the $\local$ model, the same graph $G=(V,E)$ serves both as the communication graph and the problem instance. In addition to the graph topology, the problem instance can contain local input labels. To solve the graph problem, each node is required to produce an output label so that all the labels together define a valid output.

More formally, let $\Sigma$ and $\Gamma$ be sets of input and output labels, respectively. A \emph{graph problem} is a function $\Pi_{\Sigma,\Gamma}$ that maps each graph~$G$ and input labelling $i\colon V \to \Sigma$ to a set $\Pi_{\Sigma,\Gamma}(G,i)$ of valid solutions. Each solution is a function $o\colon V \to \Gamma$. We say that \emph{algorithm~$\A$ solves graph problem $\Pi_{\Sigma,\Gamma}$} if for each graph~$G$, each input labelling $i\colon V \to \Sigma$ of $G$, and any setting of the unique identifiers, the mapping~$o\colon V \to \Gamma$ defined by setting $o(v)$ to be the local output of node~$v$ for each $v \in V$, is in the set $\Pi_{\Sigma,\Gamma}(G,i)$. Note that the unique identifiers are given as a separate labelling; the set $\Pi_{\Sigma,\Gamma}(G,i)$ of valid solutions depends only on the task-specific input labelling~$i$. When $\Sigma$ and $\Gamma$ are clear from the context, we denote a graph problem simply~$\Pi$.

Let $T\colon \N \to \N$. Suppose that algorithm~$\A$ solves problem $\Pi$, and for each input graph~$G$, each input labelling~$i$ and any setting of the unique identifiers, each node needs as most $T(|V|)$ communication rounds to halt. Then we say that \emph{algorithm~$\A$ solves problem~$\Pi$ in time~$T$}, or that the \emph{time complexity} of $\A$ is $T$. The \emph{time complexity of problem $\Pi$} is defined to be the slowest-growing function~$T\colon \N \to \N$ such that there exists an algorithm~$\A$ solving $\Pi$ in time~$T$.

In this work, we consider an important subclass of graph problems, namely \emph{locally checkable labelling} ($\lcl$) problems~\cite{naor95what}. A graph problem $\Pi_{\Sigma,\Gamma}$ is an $\lcl$ problem if the following conditions hold:
\begin{enumerate}[noitemsep]
  \item The label sets $\Sigma$ and $\Gamma$ are finite.
  \item There exists a $\local$ algorithm~$\A$ with constant time complexity, such that given any labelling $l \colon V \to \Gamma$ as an additional input labelling, $\A$ can determine whether $l \in \Pi_{\Sigma,\Gamma}(G,i)$ holds: if $l \in \Pi_{\Sigma,\Gamma}(G,i)$, all nodes output ``yes''; otherwise at least one node outputs ``no''.
\end{enumerate}
That is, an $\lcl$ problem is one where the input and output labels are of constant size, and for which the validity of a candidate solution can be checked in constant time.

\section{Link machines}\label{sec:lm}

A \emph{link machine} $M$ consists of a constant number $k$ of \emph{registers}, labelled with arbitrary strings, and a \emph{program} $P$. The program is a sequence of instructions $i_1, i_2, \dotsc, i_p$, where each instruction $i_j$ is one of the following for some registers $a$, $b$, and $c$:
\begin{itemize}[noitemsep]
    \item Addition: $a \gets b + c$.
    \item Reset: $a \gets 1$.
    \item Conditional execution: If $a = b$ (or if $a \neq b$), execute the next $s$ instructions, otherwise skip them.
\end{itemize}
The registers can store unbounded natural numbers. For convenience, we will generally identify the link machine with its program.

An \emph{execution} of the link machine $M$ is a single run through the program, modifying the values of the registers according to the instructions in the obvious way. Generally, we consider computing with link machines in a setting where
\begin{itemize}[noitemsep]
    \item all registers start from value $1$, and
    \item we are interested in the maximum value over all registers after $\ell$ executions of $M$.
\end{itemize}
Specifically, for a register $r$, we denote by $r(\ell)$ the value of register $r$ after $\ell$ full executions of the link machine program, starting from all registers set to $1$. We say that a link machine $M$ with registers $r_1, r_2, \dotsc, r_k$ has \emph{growth} $g \colon \N \to \N$ if, starting from all registers set to $1$, we have that $\max \{ r_i(\ell) \colon i = 1,2,\dotsc, k \} = g(\ell)$ for all $\ell \in \N$. While $g$ does not need to be a bijection, we use the notation $g^{-1} \colon \N \to \N$ to denote the function defined by setting $g^{-1}(\ell) = \min\set{m\in\N}{g(m) \ge \ell}$ for all $\ell \in \N$.

\subsection{Working with link machine programs}

\paragraph{Composition.} Consider two link machines $M_1$ and $M_2$ with corresponding programs $P_1$ and $P_2$. By relabelling if necessary, we can assume that the programs do not share any registers. Moreover, assume $P_1$ has a register $y$ we call the \emph{output register} for $P_1$ and $P_2$ has a register $x$ we call the \emph{input register} for $P_2$. We define the \emph{composition} $P_2 \circ P_1$ as the program
    \begin{program}
        \State{$P_1$}
        \State{$x \gets y$}
        \State{$P_2$}
    \end{program}
Note that the growth of the program $P_2 \circ P_1$ at step $\ell$ is the maximum between the growth of each program at step $\ell$, and can be affected by the input given by $P_1$ to $P_2$. The basic idea is to use this construct so that $P_1$ produces an output register $y(\ell)$ dependent on $\ell$, which is then used by $P_2$ to produce a composed growth function $g(y(\ell))$.

We define the composition $P_i \circ \dotsb \circ P_2 \circ P_1$ of multiple link machine programs with specified input and output registers similarly.

Note that the growth of a link machine is at most $2^{O(\ell)}$. As we will see later this constraint is necessary, since otherwise we would contradict known results regarding gaps on \lcl{} complexities.

\subsection{Building blocks}\label{sec:bb}

We now define our basic \emph{building blocks}, that is, small programs that can be composed to obtain more complicated functions. These building blocks are summarised in Table~\ref{tab:bb}. In all our cases, we will assume that the value of the input register $x$ is growing and bounded above by $\ell$; otherwise the semantics of a building block is undefined.

	\begin{table}
		\centering
		\begin{tabular}{@{}llll@{}}
			\toprule
			Program $P$     & Input      &  Output                    & Growth                \\
			\midrule
			\MCOUNT         & --         &  $y = \ell$                &  $\ell$               \\
			$\MROOTK'_k$     & --         &  $y = \Theta(\ell^{1/k})$  &  $\Theta(\ell^{1/k})$ \\
            $\MROOTK_k$     & $x$        &  $y = \Theta(x^{1/k})$     &  $\Theta(x)$          \\
			$\MPOWK_k$      & $x$        &  $y = \Theta(x^k)$         &  $\Theta(x^k)$        \\
			\MEXP           & $x$        &  $y = 2^{\Theta(x)}$       &  $2^{\Theta(x)}$      \\
			\MLOG           & $x$        &  $y = \Theta(\log x)$      &  $\Theta(x)$          \\
			\bottomrule
		\end{tabular}
		\caption{Our basic building blocks. The integer $k$ is a constant. Programs with no input generate output values that only depend on the number of executions $\ell$. Programs with input assume that the value of the input register $x$ is growing and bounded above by $\ell$.}\label{tab:bb}
	\end{table}

\paragraph{Link machine programming conventions.} We use the following shorthands when writing link machine programs:
\begin{itemize}
	\item We write conditional executions as \textsc{if-then} constructs, with the conditional execution skipping all enclosed instructions if the test fails. We also use \textsc{if-else} constructs, as these can be implemented in an obvious way.
	\item We write sums with multiple summands, constant multiplications, and constant additions as single instructions, as these can be easily simulated by multiple instructions and a constant number of extra registers.
\end{itemize}

\paragraph{Counting.} Our first program \MCOUNT{} simply produces a linear output $y = \ell$:
\begin{program}
	\State{$y \gets y +1$}
\end{program}
Clearly, program \MCOUNT{} has growth $\ell$.

\paragraph{Polynomials.} Next, we define a sequence of programs for computing $y = \Theta(x^k)$. For any fixed $k \ge 1$, we define the program $\MPOWK_k$ as follows:
\begin{program}
	\If{ $x \neq x_1$ }
	\State{$x_k \gets \sum_{i=0}^k \binom{k}{i} x_i$}
	\State{$x_{k-1} \gets \sum_{i=0}^{k-1} \binom{k-1}{i} x_i$}
	\State{$\dotso$}
	\State{$x_1 \gets x_1 + 1$}
	\EndIf
	\State{$y \gets x_k$}
\end{program}
We now have that $x_1 = \Theta(x)$, and by the binomial theorem, $x_i = (x_1)^i$ for all $i = 1,2,\dotsc, k$. Moreover, $\MPOWK_k$ has growth $\Theta(x^{k})$.

\paragraph{Roots.} We define two versions of a program computing a $k$th root. The first one does not take an input and has the advantage that it has sublinear growth of $\Theta(\ell^{1/k})$. Specifically, we define $\MROOTK'_k$ as follows:
\begin{program}
	\If{ $y_1 \neq y_2$ }
	\State{$y_1 \gets y_1 + 1$}
	\ElsIf{ $y_2 \neq y_3$ }
	\State{$y_2 \gets y_2 + 1$}
	\State{$y_1 \gets 1$}
	\ElsIf{ $\dotsc$ }
	\State{$\dotsc$}
	\ElsIf{ $y_{k-1} \neq y_{k}$ }
	\State{$y_{k-1} \gets y_{k-1} + 1$}
	\State{$y_{1} \gets 1, y_{2} \gets 1,  \dotsc, y_{k-2} \gets 1$}
	\Else
	\State{$y_{k} \gets y_{k} + 1$}
	\State{$y_{1} \gets 1, y_{2} \gets 1,  \dotsc, y_{k-1} \gets 1$}
	\EndIf
	\State{$y \gets y_k$}
\end{program}
Observe that started from all registers set to $1$, we always have $y_1 \le y_2 \le \ldots \le y_k$. Moreover, for register $y_k$ to increase from $s$ to $s+1$, the values of the registers $y_i$ will visit all configurations where $y_1 \le y_2 \le \ldots y_{k-1} \le s$, and there are $\binom{s+k-2}{k-1}$ such configurations. This implies that the growth of register $y_k$ is $\Theta(\ell^{1/k})$.

The second version of the $k$th root program takes an input register $x$, and computes an output $y = \Theta(x^{1/k})$. We define this program $\MROOTK_k$ as follows:
\begin{program}
	\If{ $x \neq x'$ }
	\State{$x' \gets x' + 1$}
	\State{$\MROOTK'_k$}
	\EndIf
\end{program}
Clearly, we have that $x' = \Theta(x)$, and by the properties of $\MROOTK'_k$ the output register is $y = \Theta(x^{1/k})$. The growth of $\MROOTK_k$  is $g(\ell) = x$.

\paragraph{Exponentials.} The program $\MEXP$ computes an exponential function in the input register $x$:
\begin{program}
	\If{ $x \neq x'$ }
	\State{$y \gets y + y$}
	\State{$x' \gets x' + 1$}
	\EndIf
\end{program}
We have that $x' = \Theta(x)$, and $y = 2^{x'} = 2^{\Theta(x)}$. Moreover, the growth of $\MEXP$ is $2^{\Theta(x)}$.

\paragraph{Logarithms.} The program $\MLOG$ computes a logarithm of the input register $x$:
\begin{program}
	\If{ $x \neq x'$ }
	\If{ $x' = z$ }
	\State{ $z \gets z + z $}
	\State{ $y \gets y + 1 $}
	\EndIf
	\State{$x' \gets x' + 1$}
	\EndIf
\end{program}
Clearly, we have that $x' = \Theta(x)$, and $z = \Theta(x')$. Starting from the valid starting configuration, the register $z$ only takes values that are powers of two, and $y = \log_2 z$. Thus, we have $y = \Theta(\log x)$. By construction, the growth of $\MLOG$ is $g(\ell) = z = \Theta(x)$.

\subsection{Composed functions}

By composing our building block functions, we can now construct more complicated functions, which will then be used to obtain \lcl{} problems of various complexities. The constructions we use are listed in \tableref{tab:composition}; the values of output registers and the functions computed by these programs follow directly from the results in \sectionref{sec:bb}.

Notice that for all the considered programs there is a register that is always as big as all the other registers. Thus, we can refer to it as the register of maximum value.

\begin{table}
    \centering
    \begin{tabular}{@{}lll@{}}
        \toprule
        Program $P$                                                        &             & Growth                \\
        \midrule
        $\MPOWK_p \circ \MROOTK'_q$                                        &             &  $\Theta(\ell^{p/q})$             \\
        $\MEXP \circ \MPOWK_q \circ \MROOTK'_p$                            & $(p \ge q)$ &  $2^{\Theta(\ell^{q/p})}$        \\
        $\MEXP \circ \MPOWK_q \circ \MROOTK_p \circ \MLOG \circ \MCOUNT$   & $(p \ge q)$ &  $2^{\Theta(\log^{q/p} \ell)}$        \\
        \bottomrule
    \end{tabular}
    \caption{Composed programs. The integers $p$ and $q$ are constants.}\label{tab:composition}
\end{table}

\begin{remark}
While exploring the precise power of link machines is left as an open question, we point out that, in this paper, we do not list every possible complexity that one can achieve with link machines. Indeed, there are many more time complexities that can be realised; for example, one could define a building block that performs a multiplication, or add support for negative numbers and subtractions. 
\end{remark}

\section{Link machine encoding graphs}\label{sec:link-machines}

In this section, we show how to encode link machine executions as locally checkable graphs. Fix a link machine $M$ with $k$ registers and a program of length $p$ that has non-decreasing growth $g$ in $\omega(1)$ and $2^{O(n)}$, and let $h$ be an integer. The basic idea is that we encode the link machine computation of the value $g(h)$ as follows:
\begin{itemize}
    \item We start from an $h  \times n$ grid graph, where $n=3 g(h)$, that wraps around in the horizontal direction, as shown in \figureref{fig:grid} (here $3$ is the smallest constant that avoids parallel edges or self-loops). This allows us to `count' in two dimensions; one is used for time, and the other for the values of the registers of $M$. The grid is consistently oriented so that we can locally distinguish between the two dimensions, and all grid edges are labelled with either `up', `down', `left' or `right' to make this locally checkable.
    \item We add horizontal edges to the grid graph to encode the values of the registers. Specifically, at level $\ell$ of the graph, the horizontal edges encode the values the registers take during the $\ell$th execution of the link machine program, with edge labels specifying which register values the edges are encoding (see \figureref{fig:links}).
\end{itemize}
The labels should be thought as \lcl{} input labels; as we will see later, they will allow us to recognise valid link machine encoding graphs in the sense of locally checkable proofs. We will make this construction more formal below.

\subsection{Formal definition}\label{sec:link-machine-encoding}

Let $M$ be a link machine with growth $g$ as above. We formally define the \emph{link machine encoding graphs} for $M$ as graphs obtained from the construction we describe below.

\paragraph{Grid structure.} The construction starts with a 2-dimensional $h \times n$ grid graph,
where $n = 3 g(h)$. Let $(x,\ell)$ denote the node on the $\ell$th row and the $x$th column, where $x \in \{ 1,2,\dotsc, n \}$ and $\ell \in \{ 1, 2, \dotsc, h \}$. The grid wraps around along the horizontal axis, that is, we also add the edges $\bigl((n,\ell), (1,\ell)\bigr)$ for all $\ell$.

We add horizontal \emph{link edges} to the graph according to the state of the machine $M$. That is, we say that for a node $(x,\ell)$, a \emph{link edge of length $s$} is an edge $\bigl((x,\ell), ( x+s \bmod n, \ell) \bigr)$. Let $r(\ell,i)$ denote the value of the register $r$ after executing $\ell-1$ times the full program of $M$, and then executing the first $i$ instructions of $M$. For each $\ell = 1, 2, \dotsc, h$, register $r$, and $i = 0,1,2, \dotsc, p$, we add a link edge of length $r(\ell,i)$ to all nodes on level $\ell$ if it does not already exist.

\paragraph{Local labels.} In addition to the graph structure, we add \emph{constant-size} labels to the graph as follows. First, each node $(x,\ell)$ has a set of labels for each incident edge, added according to the following rules if the corresponding edge is present (note that a single edge may have multiple labels):
\begin{itemize}[noitemsep]
    \item The grid edge to $(x,\ell+1)$ is labelled with $\dirup$.
    \item The grid edge from $(x,\ell-1)$ is labelled with $\dirdown$.
    \item The grid edge from $(x-1 \bmod n,\ell)$ is labelled with $\dirleft$.
    \item The grid edge to $(x+1 \bmod n,\ell)$ is labelled with $\dirright$.
    \item For each register $r$ and $i = 0,1,2,\dotsc,p$, the link edge of length $r(\ell,i)$ is labelled with $(r,i)$.
\end{itemize}
Consider the set of labels that each node~$(x,\ell)$ associates to each of its incident edges. When we later define graph problems, we assume these labels to be implicitly encoded in the node label given to $(x,\ell)$.

\begin{figure}
\centering
\includegraphics[page=1]{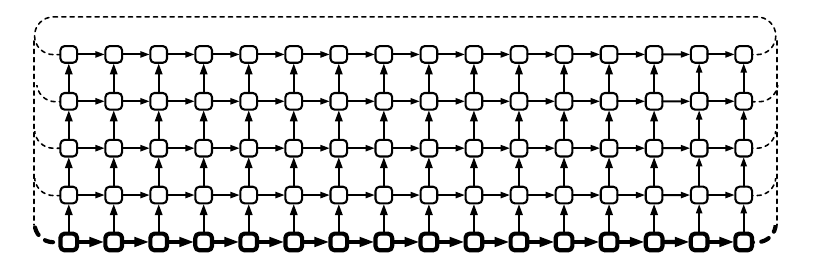}
\caption{The base grid. The bottom cycle (first level) is highlighted.}\label{fig:grid}
\end{figure}

\paragraph{Input.}
Also, each node $v$ is provided with an input $i(v)  \in \{0,1\}$.

\subsection{Local checkability}\label{subsec:local-checkability}
We show that the labels described in Section~\ref{sec:link-machine-encoding} constitute a \emph{locally checkable proof} for the graph being a link machine encoding graph. That is, there is a \local{} algorithm where all nodes accept a labelled graph if and only if it is a link machine encoding graph.

\paragraph{Local constraints.} We first specify a set of local constraints that are checked by the nodes locally. All the constraints depend on the radius-$4$ neighbourhood of the nodes, so this can be implemented in the \local{} model in $4$ rounds. In the following, for labels $L_1, L_2, \dotsc, L_k$, let $v(L_1,L_2,\dotsc,L_k)$ denote the node reached by following the edges with the specified labels.
The full constraints are now as follows:

\begin{enumerate}
	\item\label{item:link-first} Each node checks that the basic properties of the labelling are correct:
	\begin{itemize}
		\item All possible edge labels are present exactly once, except possibly one of $\dirdown$ and $\dirup$.
		\item The direction labels $\dirup$, $\dirdown$, $\dirleft$, and $\dirright$ are on different edges if present.
	\end{itemize}
	\item Grid constraints ensure the validity of the grid structure:
	\begin{itemize}
		\item Each node checks that each of the edges labelled with $\dirup$, $\dirdown$, $\dirleft$ or $\dirright$ has the opposite label in the other end.
		\item If there is an edge labelled \dirdown, check that $v(\dirdown,\dirright,\dirup) = v(\dirright)$.
		\item If there is not an edge labelled \dirdown, check that also nodes $v(\dirleft)$ and $v(\dirright)$ do not have edges labelled \dirdown.
		\item If there is not an edge labelled \dirup, check that also nodes $v(\dirleft)$ and $v(\dirright)$ do not have edges labelled \dirup.
	\end{itemize}
	\item Nodes check that the values of the registers are correctly initialised on the link edges:
	\begin{itemize}
		\item Nodes that do not have an edge labelled with $\dirdown$ check that the register values are initialised to $1$, that is, the labels $\dirright$ and $(r,0)$ are on the same edge for all registers $r$.
		\item Nodes that have an edge labelled with $\dirdown$ check that the registers are copied correctly, that is, $v((r,0)) = v(\dirdown,(r,p),\dirup)$ for all registers $r$.
	\end{itemize}
	\item Nodes check that the program execution is encoded correctly as follows. Each instruction is processed in order, from $1$ to $p$. The $i$th is checked as follows, depending on the type of the instruction:
	\begin{itemize}
		\item If the instruction is $a \gets 1$:
		\begin{enumerate}
			\item Register $a$ is correctly set to $1$: the labels $\dirright$ and $(a,i)$ are on the same edge.
			\item Any of the other registers did not change, that is, labels $(r,i-1)$ and $(r,i)$ are on the same edge for all registers $r$ except $a$.
		\end{enumerate}
		\item If the instruction is $a \gets b + c$:
		\begin{enumerate}
			\item Register $a$ is set correctly: $v\bigl((a,i)\bigr) = v\bigl((b,i-1), (c,i-1)\bigr)$.
			\item Any of the other registers did not change, that is, labels $(r,i-1)$ and $(r,i)$ are on the same edge for all registers $r$ except $a$.
		\end{enumerate}
		\item If the instruction is an \textsc{if} statement comparing registers $a$ and $b$, check if the labels $(a,i-1)$ and $(b,i-1)$ are on the same edge, and if this does not match the condition of the \textsc{if} statement, check that the following $s$ instructions are not executed:
		\begin{enumerate}
			\item Any registers do not change for $s$ steps, that is, for all registers $r$, we have that labels $(r,i-1), (r,i), (r,i+1), \dotsc, (r,i+s)$ are on the same edge.
			\item Skip the checks for the next $s$ instructions.
		\end{enumerate}
	\end{itemize}
	\item\label{item:link-last} If no edges are labelled \dirup, check that the link edges corresponding to the register with the maximum value form $3$-cycles.
\end{enumerate}

\begin{figure}
\centering
\includegraphics[page=2]{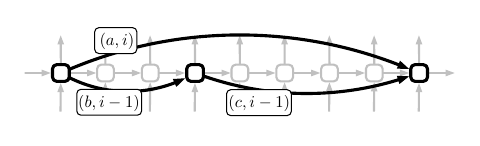}
\caption{Checking the correct encoding of the execution of the instruction $a \gets b + c$.}\label{fig:links}
\end{figure}

\paragraph{Correctness.} It is clear that link machine encoding graphs satisfy the constraints \ref{item:link-first}--\ref{item:link-last} specified above. Conversely, we want to show any graph satisfying these constraints is a link machine encoding graph, but it turns out this is not exactly the case.

It might happen that the register values exceed the width $w$ of the grid, the edges ``wrap around'', and the correspondence between the edge lengths and register values gets lost. However, for this to happen one has to have a row $\ell$ with $g(\ell) \ge w$.

In order to characterise the graph family captured by the local constraints, we define that a graph $G$ is an \emph{extended} link machine encoding graph if
\begin{itemize}[noitemsep]
	\item $G$ is an $h \times w$ grid for some $h$ and $w$ that wraps around horizontally but not vertically,
	\item $G$ satisfies the \emph{local} constraints of link machine encoding graphs, and
	\item there is an $\ell \le h$ with $g(\ell) \ge w/3$ such that up to row $\ell-1$ the edge lengths are in one-to-one correspondence with the register values of the first $\ell-1$ executions of link machine~$M$.
\end{itemize}
Note that a link machine encoding graph is trivially an extended link machine encoding graph, as we can simply choose $\ell = h$ and $w = n$ and hence $g(\ell) = w/3$. The intuition is that extended link machine encoding graphs have a good bottom part with dimensions $\ell \times \Theta(g(\ell))$, and on top of that there might be any number of additional rows of some arbitrary garbage.

\begin{lemma}\label{lem:extendedmachines}
	Let $G$ be a graph that satisfies the constraints \ref{item:link-first}--\ref{item:link-last} and has at least one node that does not have an edge labelled with $\dirdown$. Then $G$ is an extended link machine encoding graph.
\end{lemma}

\begin{proof}
	By constraints~(1) and (2), we have that the graph $G$ is a grid graph that wraps around horizontally. By the assumption that there is a node without edge labelled by $\dirdown$ and by constraint~(2), the grid cannot wrap around vertically. Hence $G$ is an $h \times w$ grid that wraps around horizontally, for some values $h$ and $w$, and by assumption it satisfies the local constraints of link machine encoding graphs.

	Constraints~(3) and (4) ensure that the link edges and the corresponding labels are according to the link machine encoding graph specification, as long as $g(\ell) \le w$.

	Constraint~(5) ensures that we cannot have $g(\ell) < w/3$ for all $\ell$, as in the top row we must have three edges that form a cycle that wraps around the entire grid of width $w$ at least once. Hence at some point we must reach $g(\ell) \ge w/3$, and this is sufficient for $G$ to satisfy the definition of an extended link machine encoding graph.
\end{proof}

\section{\boldmath \texorpdfstring{\lcl{}}{LCL} constructions}

Let $M$ be a link machine with non-decreasing growth $g$ in $\omega(1)$ and $2^{O(n)}$, and let $\Pi$ be an \lcl{} problem on directed cycles with complexity $T(n)$ -- for concreteness, $\Pi$ will either be $3$-colouring (complexity $\Theta(\log^* n)$) or a variant of $2$-colouring (complexity $\Theta(n)$). To simplify the construction, we will assume that $\Pi$ is solvable on directed cycles with \emph{one-sided} algorithms, i.e., with algorithms in which each node only looks at its $T(n)$ successors.
We now construct an \lcl{} problem $\Pi_M$ with complexity related to $g$, as outlined in the introduction:
\begin{itemize}
    \item If a node sees a graph that locally looks like a link machine encoding graph for $M$, and the node is on the bottom row of the grid, it will need to solve problem $\Pi$ on the directed cycle formed by the bottom row of the grid. As will be shown later, in $\Theta(\ell)$ steps, a node on the bottom row of the grid sees all nodes within distance $f(\ell) = \ell g(\ell)$ on the bottom cycle, so this is solvable in $\Theta(f^{-1}(T(n)))$ rounds.
    \item If a node sees something that does not look like a link machine encoding graph, it is allowed to report an error; the node must also provide an \emph{error pointer} towards an error it sees in the graph. A key technical point here is to ensure that it is not too easy to claim that there is an error somewhere far, even if in reality we have a valid link machine encoding graph. 
    We address this by ensuring that error pointer chains can only go right and then up, they cannot disappear without meeting an error, the part that is pointing right must be properly $2$-coloured, and the part that is pointing up copies the input $i(v)$ given to the node $v$ that is witnessing the error. If some nodes claim that the error is somewhere far up, we will eventually reach the highest layer of the graph and catch the cheaters. Also, nodes cannot blindly point up, because they need to mark themselves with the input of the witness. If all bottom-level nodes claim that the error is somewhere right, we do not necessarily catch cheaters, but the nodes did not gain anything as they had to produce a proper $2$-colouring for the long chain of error pointers.
\end{itemize}
There are some subtleties in both of these points, which we address in detail below.

\subsection{\boldmath The \texorpdfstring{$\lcl$}{LCL} problem \texorpdfstring{$\Pi_M$}{Pi\_M}}

Formally, we specify the \lcl{} problem $\Pi_M$ as follows. The input label set for $\Pi_M$ is the set of labels used in the link machine encoding graph labelling for $M$ as described in \sectionref{sec:link-machines}. The possible output labels are the following:
\begin{enumerate}[noitemsep]
    \item output labels of the \lcl{} problem $\Pi$,
    \item an \emph{error label} $\error$,
    \item an \emph{error pointer}, pointing either right ($\dirright$) or up ($\dirup$), with a counter mod $2$ and a label $c_e$ in $\{0,1\}$,
    \item an empty output $\epsilon$.
\end{enumerate}
The correctness of the output labelling is defined as follows.
\begin{enumerate}
    \item If the input labelling does not locally satisfy the constraints of link machine encoding graphs for $M$ (see Section~\ref{subsec:local-checkability}) either at the node itself or at one of its neighbours, then the only valid output is \error{}. Otherwise, the node must produce one of the other labels.
    \item If the output of a node $v$ is one of the labels of $\Pi$, then the following must hold:
    \begin{itemize}
        \item Node $v$ does not have an incident edge with label $\dirdown$ in the input labelling.
        \item If any adjacent nodes have output from the label set of $\Pi$, then the local constraints of $\Pi$ must be satisfied.
    \end{itemize}
    \item If the output of a node $v$ is empty, then it must have an incident edge with label $\dirdown$ in the input labelling.
    \item If the output of a node $v$ is an error pointer, then the following must hold:
    \begin{itemize}
        \item Node $v$ has only one outgoing error pointer.
        \item The error pointer is pointing either $\dirright$ or $\dirup$ if the node does not have an edge labelled $\dirdown$ in the input, and $\dirup$ if the node does have an edge labelled $\dirdown$ in the input.
        \item The node at the other end of the pointer either outputs an error label $\error$, or an error pointer.
        \item The $\bmod$-$2$ counters of nodes outputting error pointer $\dirright$ form a 2-colouring in the induced subgraph of those nodes.
        \item The nodes~$v$ outputting error pointer $\dirup$ have the same label $c_e(v)$ as the next node in the chain. If $v$ is the last node in the chain, $c_e(v) = i(w)$ holds, where $w$ is the witness outputting $\error$.
    \end{itemize}
\end{enumerate}
These conditions are clearly locally checkable, so $\Pi_M$ is a valid \lcl{} problem.

\subsection{Time complexity}

We now prove the following bounds for the time complexity of problem $\Pi_M$; recall that $f(k) = k g(k)$, where $g$ is the growth of link machine $M$. In the following, $n$ denotes the number of nodes in the input graph, and $\hat{n}$ is the smallest number satisfying $n \le \hat{n} \cdot f^{-1}(T(\hat{n}))$. The intuition here is that the ``worst-case instances'' of size $n$ will be grids of width approximately $\hat{n}$ and height approximately $f^{-1}(T(\hat{n}))$.

\begin{restatable}{theorem}{thmupperbound}\label{thm:upper-bound}
	Problem $\Pi_M$ can be solved in $O\bigl(f^{-1}\bigl(T(\hat{n})\bigr)\bigr)$ rounds.
\end{restatable}

\vspace{-2ex}
\begin{restatable}{theorem}{thmlowerbound}\label{thm:lower-bound}
	Problem $\Pi_M$ cannot be solved in $o\bigl(f^{-1}\bigl(T(\hat{n})\bigr)\bigr)$ rounds.
\end{restatable}

\subsubsection{Upper bound -- proof of Theorem \ref{thm:upper-bound}}\label{app:upper-bound}

We start by observing that the link machine encoding graphs essentially provide a `speed-up' in terms of how quickly the nodes on the bottom cycle can see other nodes on the bottom cycle. Recall that $f(k) = k g(k)$, where $g$ is the growth of link machine $M$; also recall that $g(k) = 2^{O(k)}$.

\begin{lemma}\label{lemma:link-speedup}
	Let $G$ be a link machine encoding graph for $M$, and let $u$ and $v$ be nodes on the bottom cycle. If $u$ can be reached from $v$ in $\ell$ steps following edges labelled with $\dirright$, then $u$ can be reached from $v$ in $O\bigl(f^{-1}(\ell)\bigr)$ steps following edges labelled with $\dirup$, $\dirdown$, $\dirright$ or register labels.
\end{lemma}

\begin{proof}

	Starting from a node $v$ on the cycle, in $O(k)$ steps it is possible to see a node on the cycle that is $f(k) = k g(k)$ steps away: take $k$ steps up, $k$ steps right along shortcuts, and $k$ steps down. We will use a similar procedure to go to a node at any distance $\ell$.
	
	Let $k$ be the smallest value such that $f(k) \ge \ell$. By assumption, $f(k-1) = (k-1)g(k-1) < \ell$. Since $g(k) = 2^{O(k)}$, we have $g(k)/g(k-1) \le k-1$ for large enough graphs, and hence $g(k) \le (k-1)g(k-1) = f(k-1) < \ell$.
	
	We find a path $P$ from $v$ to $u$ by a greedy procedure. First go up for $k$ steps. Recall that at height $k$ there are shortcuts of length $g(k)$. Go right along shortcuts until the distance to the column of $u$ is less than $g(k)$ (taking one more shortcut would bring us to a node that is on the right of~$u$). This takes at most $k$ steps. Next step down and do a greedy descent to get to the column of~$u$. At each level $h$, if the remaining distance to the column of $u$ is at least $g(h)$, take steps along the longest shortcut until the distance is less than $g(h)$. Since $g(k) = 2^{O(k)}$, the length of the shortcuts at level $h$ is at most a constant number of times the length at level $h-1$, hence this number of steps is bounded by $O(1)$. Finally, step down. We either reach the column of $u$ or the bottom cycle, and in this case the distance to the column of $u$ is less than $g(1)$.
	
    Since $f(k-1) < \ell$, then $k < f^{-1}(\ell) +1$. We take a total of at most $k$ steps up and down, at most $k$ steps right at level $k$, and $O(1)$ steps for each of the $k$ levels, for a total of $O(f^{-1}(\ell))$ steps.
\end{proof}

\begin{lemma}\label{lemma:hatn}
	If a node not having an edge labelled \dirdown{} sees no errors within distance $r = C f^{-1}(T(\hat{n}))$ for a sufficiently large constant $C$, then it can produce a valid output for the problem $\Pi$.
\end{lemma}

\begin{proof}
	First consider a global problem, i.e., $T(n) = \Theta(n)$. If we explore the grid up and right for $r = C f^{-1}(\hat T(n))$ steps, and we do not encounter any errors, and the grid does not wrap around, then we would discover a grid fragment of dimensions at least $D f^{-1}(\hat n) \times D \hat n$ for a $D = \Omega(C)$. Such a grid fragment would contain $D^2 n$ nodes, and for a sufficiently large $D$ this would contradict the assumption that the input has $n$ nodes. Hence we must encounter errors (which by assumption is not the case), or the grid has to wrap around cleanly without any errors, in which case we also see the entire bottom row and we can solve $\Pi$ there by brute force.

	Second, consider the case $T(n) = \Theta(\log^* n)$. By a similar reasoning, the node can gather a grid fragment of dimensions $D f^{-1}(\log^* \hat n) \times D \log^* \hat n$. In particular, it can see a fragment of length $D \log^* \hat n$ of the bottom row. Furthermore, we have $\log^* n = O(\log^* \hat{n})$: to see this, note that $g$ is non-decreasing, $f(k) \ge k$, and hence $n = o(\hat{n}^2)$. Therefore in $r$ rounds, for a sufficiently large $C$, we can gather a fragment of the bottom row that spans up to distance at least $T(n)$, and this is enough to solve~$\Pi$.
\end{proof}

As a consequence, we obtain an upper bound for the complexity of $\Pi_M$.

\thmupperbound*

\begin{proof}
	The idea of the algorithm that solves the described \lcl{} problem is the following.
	First, each node gathers its constant-radius neighbourhood, and sees if there is a local error:
	\begin{itemize}
		\item If a node witnesses a local error, it marks itself as `witness'
		\item If a node is either a witness itself, or it is adjacent to a witness, it marks itself as a \emph{near-witness}, and outputs $\error$.
	\end{itemize}
	
	Now let $r = c\, f^{-1}(T(\hat{n}))$, for a large enough constant $c$. Each node $v$ in the bottom cycle -- not having an edge labelled \dirdown{} -- attempts to gather full information about an $r \times f(r)$ rectangle to the up and right from node $v$, that is, a rectangle composed by the bottom-most $r$ nodes of the first $f(r)$ columns to the right of $v$. By \lemmaref{lemma:link-speedup}, in $O(r)$ rounds we can either successfully gather the entire rectangle if it is error-free, or we can discover the nearest column that contains a near-witness:
	\begin{enumerate}
		\item If the entire rectangle is error-free, we can solve $\Pi$ on the bottom row by \lemmaref{lemma:hatn}. 
		\item Otherwise, we find the nearest column containing a near-witness $\bar{w}$. In such a case, node $v$ will output its modulo-2 distance from that column, the input $i(\bar{w})$ of the witness, and produce a path of error pointers that spans a sequence of edges labelled $\dirright$ followed by a sequence of edges labelled $\dirup$, reaching $\bar{w}$. Notice that this path is unique and always exists, since all columns before the nearest one containing a witness must be fault-free (up to height $r$), and if a witness is in the same column of $v$, the lowest one can be always reached by a fault-free path spanning only edges labelled $\dirup$.
	\end{enumerate}
	Finally, nodes that are not on the bottom cycle and do not see bottom nodes wanting to produce error pointer paths produce empty outputs.
	
	Clearly, this produces a valid solution to $\Pi_M$ on extended link machine encoding graphs, since they satisfy \lemmaref{lemma:hatn}. Also, if there are no witnesses and every node has an edge labelled \dirdown{}, all nodes produce empty outputs, that is valid. 
	
	Now, consider a graph that is not an extended link machine encoding graph. A node will explore the graph for $\Theta(r)$ rounds. If the node satisfies the requirements of \lemmaref{lemma:hatn}, then it produces a valid solution for the problem $\Pi$. Otherwise the node sees a witness. If a node $v$ decides to produce an error pointer towards a near-witness $\bar{w}$, then all the nodes on the error path will produce an error pointer towards $\bar{w}$. This follows from the observation that, on valid fragments, nodes on the same row reach the same height while visiting the graph, due to the rectangular visit. Thus, if $v$ outputs a pointer towards $\bar{w}$, then all the intermediate nodes will output a pointer, and these pointers will correctly produce a path from $v$ to $\bar{w}$ with the right modulo-2 distance and $c_w$ labelling.
\end{proof}

\subsubsection{Lower bound -- proof of Theorem \ref{thm:lower-bound}}\label{app:lower-bound}

Next, we prove that the upper bound in \theoremref{thm:upper-bound} is tight. The worst-case instances are \emph{truncated} link machine encoding graphs, defined as follows: take a valid link machine encoding graph and remove rows from the top, until it is satisfied that $n \le \hat{n} \cdot f^{-1}(T(\hat{n}))$, where $\hat{n}$ is the length of the bottom cycle. The basic idea is to show that on truncated link machine encoding graphs, any algorithm has two choices, both of them equally bad:
\begin{itemize}[noitemsep]
	\item We can solve problem $\Pi$ on the bottom cycle, but this requires time $\Omega\bigl(f^{-1}(T(\hat{n}))\bigr)$.
	\item We can report an error, but this also requires time $\Omega\bigl(f^{-1}(T(\hat{n}))\bigr)$.
\end{itemize}
Note that truncated link machine encoding graphs have errors, and hence it is fine for a node to report an error. However, all witnesses are on the top row (or next to it), and constructing a correctly labelled error pointer chain from the bottom row to the top row takes time linear in the height of the construction. We will formalise this intuition in what follows.

\begin{lemma} \label{lem:lm-view}
	Let $G$ be a truncated link machine encoding graph,  $v$ be a node on the bottom row, and let $h$ be any function satisfying $h(\ell) = o\bigl(f^{-1}(\ell)\bigr)$. Let $X$ be the set of all nodes that $v$ can see in $h(\ell)$ steps. Then $X$ is contained in the subgraph induced by the columns within distance $o(\ell)$ of $v$.
\end{lemma}

\begin{proof}
	By the construction of the link machine encoding graphs, the maximum distance in columns we can reach in $\ell$ steps is bounded by $f(\ell)$. Since $f(\ell) = \omega(\ell)$, we have that $f\bigl( \ell/C \bigr) \le f(\ell)/C$ for any positive integer $C$. Thus, for any $C$, there is $\ell_0$ such that $h(\ell) < f^{-1}(\ell)/C$ for all $\ell \ge \ell_0$, and thus
	\[ f(h(\ell)) < f\bigl( f^{-1}(\ell)/C\bigr) \le \ell/C\]
	for all $\ell \ge \ell_0$, which implies the claim.
\end{proof}

\thmlowerbound*

\begin{proof}
	The nodes on the bottom row have the following possible outputs:
	\begin{enumerate}
		\item At least one node $v$ produces an error pointer $\dirup$. Then we must have a chain of $\dirup$ pointers all the way to the near-witness $w$ near the top row, and the chain has to be labelled with the input of $w$. The distance from $v$ to $w$ is $\Theta(f^{-1}(T(\hat{n})))$, and the claim follows.
		\item None of the nodes on the bottom row produce an error pointer $\dirup$, but at least one of them produces an error pointer $\dirright$. But then all nodes on the bottom row must output $\dirright$, and the bottom cycle has to be properly $2$-coloured.
		\item None of the nodes produce any error pointers. Then all nodes on the bottom row must solve problem $\Pi$.
	\end{enumerate}
	As $2$-colouring the bottom row is at least as hard as solving problem $\Pi$ on the bottom row, it is sufficient to argue that the third case requires $\Omega(f^{-1}(T(\hat{n})))$ rounds. The proof is by simulation. We assume a faster algorithm for $\Pi_M$ and use it to speed up the corresponding problem $\Pi$ on cycles.
	
	Let $A$ be an algorithm for $\Pi_M$ with running time $o(f^{-1}(T(\hat{n})))$. The algorithm has to solve the problem $\Pi$ on the bottom cycle. Now, given a cycle $C$ of length $\hat{n}$ as input, we create a virtual link machine encoding graph on top of the cycle as follows: each node creates the nodes in its column, their identifiers defined to be the identifier of the bottom node padded with the node's height, encoded in $\log h$ bits.
	
    To simulate an algorithm with running time $t = o(f^{-1}(T(\hat{n})))$ in this virtual graph, each node needs to learn the identifiers of all nodes in its radius-$t$ neighbourhood in the virtual graph. By Lemma~\ref{lem:lm-view}, the columns of those nodes are contained within distance $o(T(\hat{n}))$ in the virtual graph. Thus, we can recover the identifiers of the nodes by scanning the cycle $C$ up to distance $o(T(\hat{n}))$. Now each node $v$ can apply $A$ and find a solution for $\Pi$ on the cycle $C$ in time $o(T(\hat{n}))$, by outputting the output of the node at the bottom of the virtual column created by node $v$. This yields an algorithm with running time $o(T(\hat{n}))$ on cycles of length $\hat{n}$, a contradiction.
\end{proof}

\begin{remark}
	Note that we did not use the fact that our algorithms are deterministic in this proof. In fact, a similar argument can be applied to randomised algorithms. This is due to the fact that, as we will see later, we consider problems that are equally hard for randomised and deterministic algorithms. Also, on truncated link machine encoding graphs the only way to cheat with error pointers is to produce a 2-colouring or copy the input of nodes that are far on the graph, that is, to solve problems that are equally hard for randomised and deterministic algorithms.
\end{remark}

\subsection{\boldmath Instantiating the \lcl{} construction}
We consider the problems $\Pi_g$ and $\Pi_l$ defined on cycles as follows:
\begin{description}
	\item[\boldmath $\Pi_l$ (\emph{$3$-colouring}):] Output a proper $3$-colouring.

	The time complexity of this problem is $\Theta(\log^* n)$ \cite{linial92locality,cole86deterministic}, and it can be solved with one-sided algorithms. This is also clearly an $\lcl$ problem.
	\item[\boldmath $\Pi_g$ (\emph{safe $2$-colouring}):] Given an input in $\{0,1\}$, label the nodes with $\{0,1,2,\error\}$ such that
	\begin{itemize}[noitemsep]
		\item input-$0$ nodes are labelled with $0$,
		\item input-$1$ nodes are labelled with $1$, $2$, or $\E$,
		\item $1$ is never adjacent to $1$,
		\item $2$ is never adjacent to $2$,
		\item $\error$ is never adjacent to $0$, $1$, or $2$.
	\end{itemize}
	In essence, if we have an all-$1$ input, we can produce an all-$\error$ output, and if we have an all-$0$ input, we can produce an all-$0$ output. However, if we have a mixture of $0$s and $1$s, we must properly $2$-colour each contiguous chain of $1$s. The worst-case instance is a cycle with only one $0$, in which case we must properly $2$-colour a chain of length $n-1$.

	The time complexity of this problem is $\Theta(n)$, and it can be solved with one-sided algorithms. It is also clearly an $\lcl$ problem. Note that, unlike $2$-colouring, safe $2$-colouring is always solvable for any input (including odd cycles).
\end{description}

We now instantiate our \lcl{} construction using the link machines defined in \sectionref{sec:lm}. The  general recipe of these instantiations will be the following:
\begin{itemize}
    \item We start with a link machine program $M$ with growth $g$, and compute the function $f(\ell) = \ell g(\ell)$ that controls the speed-up.
    \item Next, we observe that there can be link machine encoding graphs with $n$ nodes, in which the bottom cycle has length $\hat{n}$ satisfying  $n = \Theta\bigl(\hat{n} \cdot f^{-1}(T(\hat{n}))\bigr)$, in which nodes of the bottom cycle see no errors within distance $\Theta(f^{-1}(T(\hat{n})))$.
    \item Now it follows from Theorems~\ref{thm:upper-bound} and \ref{thm:lower-bound} that when we instantiate the construction, we get problems of complexity
    \begin{enumerate}
        \item $T_1(n) = \Theta\bigl(f^{-1}(\hat{n})\bigr)$ when starting from $\Pi_g$, and
        \item $T_2(n) = \Theta\bigl(f^{-1}( \log^* \hat{n})\bigr)$ when starting from $\Pi_l$.
    \end{enumerate}
\end{itemize}
By considering each of the composite functions of Table~\ref{tab:composition}, and by applying Theorems~\ref{thm:upper-bound} and \ref{thm:lower-bound}, we obtain all of the new time complexities listed in Table~\ref{tab:contrib}. 

\begin{theorem}
	There exist \lcl{} problems of complexities
	\begin{enumerate}[noitemsep]
		\item $\Theta\bigl(n^{r/s}\bigr)$,
		\item $\Theta\bigl((\log^* n)^{q/p}\bigr)$,
	\end{enumerate}
	where $r$,$s$, $p$ and $q$ are positive integer constants, satisfying $q/p\le1$ and $r/s < 1/2$.
\end{theorem}
\begin{proof}
	Let $M = \MPOWK_p \circ \MROOTK'_q$ with growth $g(\ell) = \Theta(\ell^{p/q})$. We have
	\begin{itemize}[noitemsep]
		\item $f(\ell) = \Theta\bigl(\ell^{(p+q)/q}\bigr)$, and 
		\item $f^{-1}(x) = \Theta\bigl(x^{q/(p+q)}\bigr)$.
	\end{itemize}
    When $\Pi$ is $\Pi_g$ we obtain: 
	\begin{itemize}[noitemsep]
		\item $n = \Theta(\hat{n}^{(p+2q)/(p+q)})$
		\item $\hat{n} = \Theta\bigl(n^{(p+q)/(p+2q)}\bigr)$
		\item $T_1(n) = \Theta\bigl(\hat{n}^{q/(p+q)}\bigr) = \Theta\bigl(n^{q/(p+2q)}\bigr)$.
	\end{itemize}
    By setting $q=r$ and $p = s- 2r$ the claim follows.\\
    When $\Pi$ is $\Pi_l$ we obtain:
	\begin{itemize}[noitemsep]
		\item $\hat{n} = \widetilde\Theta(n)$
		\item $T_2(n) = \Theta\bigl((\log^* n)^{q/(p+q)}\bigr)$.
	\end{itemize}
	The claim follows by setting the values of $p$ and $q$ appropriately.
\end{proof}

\begin{theorem}
	There exist \lcl{} problems of complexities
	\begin{enumerate}[noitemsep]
		\item $\Theta( \log^{p/q} n )$, and
		\item $\Theta( \log^{p/q} \log^* n )$,
	\end{enumerate}
	where $p$ and $q$ are positive integer constants such that $p/q \ge 1$.
\end{theorem}

\begin{proof}
	Let $M = \MEXP \circ \MPOWK_q \circ \MROOTK'_p$ with growth $g(\ell) = 2^{\Theta(\ell^{q/p})}$. We have
	\begin{itemize}[noitemsep]
		\item $f(\ell) = 2^{\Theta(\ell^{q/p})}$ and $f^{-1}(x) = \Theta(\log^{p/q} x)$,
		\item $\hat{n} = \widetilde\Theta(n)$.
	\end{itemize}
	Thus, the \lcl{} problem $\Pi_M$ has complexity
	\begin{itemize}[noitemsep]
		\item $T_1(n) = \Theta\bigl( \log^{p/q} \hat{n} \bigr) = \Theta\bigl( \log^{p/q} n \bigr)$ when $\Pi$ is $\Pi_g$, and
		\item $T_2(n) = \Theta\bigl( \log^{p/q} \log^* \hat{n} \bigr) = \Theta\bigl( \log^{p/q} \log^* n \bigr)$ when $\Pi$ is $\Pi_l$.
		\qedhere
	\end{itemize}
\end{proof}

\begin{theorem}
	There exist \lcl{} problems of complexities
	\begin{enumerate}[noitemsep]
		\item $2^{\Theta( \log^{q/p} n )}$, and
		\item $2^{\Theta( \log^{q/p} \log^* n )}$,
	\end{enumerate}
	where $p$ and $q$ are positive integer constants such that $q/p\le1$.
\end{theorem}
\begin{proof}
	Let $M = \MEXP \circ \MPOWK_p \circ \MROOTK_q \circ \MLOG \circ \MCOUNT$ with growth $g(\ell) = 2^{\Theta(\log^{p/q} \ell)}$. We have
	\begin{itemize}[noitemsep]
		\item $f(\ell) = 2^{\Theta(\log^{p/q} \ell)}$ and $f^{-1}(x) = 2^{\Theta(\log^{q/p} x)}$,
		\item $\hat{n}$ is $\Omega(n^{1/2})$ and $O(n)$.
	\end{itemize}
	Thus, the \lcl{} problem $\Pi_M$ has complexity
	\begin{itemize}[noitemsep]
		\item $T_1(n) = 2^{\Theta(\log^{q/p} \hat{n} )} = 2^{\Theta( \log^{q/p} n )}$ when $\Pi$ is $\Pi_g$, and
		\item $T_2(n) = 2^{\Theta(\log^{q/p} \log^* \hat{n} )} = 2^{\Theta( \log^{q/p} \log^* n )}$ when $\Pi$ is $\Pi_l$.
		\qedhere
	\end{itemize}
\end{proof}

\section*{Acknowledgements}

We thank Christopher Purcell for discussions, Sebastian Brandt for spotting an error in the preliminary version of the paper, and the anonymous reviewers for their useful comments.
This work was supported in part by the Academy of Finland, Grant 285721, the Ulla Tuominen Foundation, and ANR Project DESCARTES.

\DeclareUrlCommand{\Doi}{\urlstyle{same}}
\renewcommand{\doi}[1]{\href{http://dx.doi.org/#1}{\footnotesize\sf doi:\Doi{#1}}}

\bibliographystyle{plainnat}
\bibliography{lcl-complexity}

\end{document}